\documentclass{llncs}

\usepackage{graphicx}

\usepackage[colorlinks=false]{hyperref}

\usepackage[margin=1in]{geometry}

\usepackage{amsthm}
\usepackage{amsmath}
\usepackage{amssymb}

\newcommand{\defparproblem}[5]{
	\vspace{1mm}
	\noindent\fbox{
		\begin{minipage}{0.96\textwidth}
			\begin{tabular*}{\textwidth}{@{\extracolsep{\fill}}lr} \textsc{#1} & {\bf{Parameter:}} #3 \\ \end{tabular*}
			{\bf{Input:}} #2 \\
			{\bf{#5:}} #4
		\end{minipage}
	}
	\vspace{1mm}
}

\let\origqed\qed
\newcommand{\claimqed}{\hfill$\lrcorner$}
\newenvironment{claimproof}[1][\proofname]{\begin{proof}\renewcommand{\qed}{\claimqed}}{\end{proof}\renewcommand{\qed}{\origqed}}

\spnewtheorem{subclaim}{Claim}{\itshape}{\rmfamily}

\newcommand{\cc}{\mathrm{cc}}
\newcommand{\Oh}{\mathcal{O}}
\newcommand{\hh}{\mathcal{H}}
\newcommand{\tw}{\mathrm{\textbf{tw}}}
\newcommand{\hhtw}[1][\hh]{\tw_{#1}}
\newcommand{\hhdepth}[1][\hh]{\ed_{#1}}

\newcommand{\td}{\mathrm{\textbf{td}}}
\newcommand{\ed}{\mathrm{\textbf{ed}}}
\newcommand{\torso}{\textbf{T}}
\newcommand{\bip}{\mathrm{bip}}

\newcommand{\adj}{\textbf{adj}}
\newcommand{\inc}{\textbf{inc}}

\ifdefined\DEBUG{}

\newcommand{\bmpr}[1]{\rem{\textcolor{red}{\(\bullet \) #1}}}
\newcommand{\jjh}[1]{{\color{orange}{#1}}}

\else

\newcommand{\bmpr}[1]{}
\newcommand{\jjh}[1]{#1}
\fi
\newcommand{\depr}[1]{}

\date{}

\title{FPT Algorithms to Compute the Elimination Distance to Bipartite Graphs and More\thanks{This project has received funding from the European Research Council (ERC) under the European Union's Horizon 2020 research and innovation programme (grant agreement No 803421, ReduceSearch).}}

\titlerunning{FPT Algorithms to Compute Elimination Distance}
\author{Bart M.P. Jansen \and Jari J.H. de Kroon}
\date{\today}

\institute{Eindhoven University of Technology}

\begin{document}

\maketitle

\begin{abstract}
For a hereditary graph class~$\hh$, the~$\hh$-elimination distance of a graph~$G$ is the minimum number of rounds needed to reduce~$G$ to a member of~$\hh$ by removing one vertex from each connected component in each round. The $\hh$-treewidth of a graph~$G$ is the minimum, taken over all vertex sets~$X$ for which each connected component of~$G - X$ belongs to~$\hh$, of the treewidth of the graph obtained from~$G$ by replacing the neighborhood of each component of~$G-X$ by a clique and then removing~$V(G) \setminus X$. These parameterizations recently attracted interest because they are simultaneously smaller than the graph-complexity measures treedepth and treewidth, respectively, and the vertex-deletion distance to~$\hh$. For the class~$\hh$ of bipartite graphs, we present non-uniform fixed-parameter tractable algorithms for testing whether the $\hh$-elimination distance or $\hh$-treewidth of a graph is at most~$k$. Along the way, we also provide such algorithms for all graph classes~$\hh$ defined by a finite set of forbidden induced subgraphs.
\newline \includegraphics[scale=0.1]{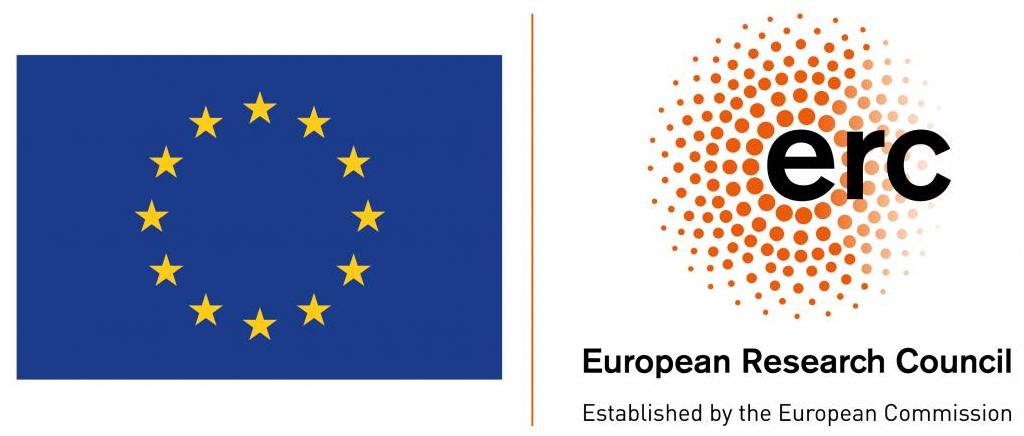}
\end{abstract}

\section{Introduction}

\paragraph{Background}
Assuming some structure on the input of a computational problem can greatly decrease its difficulty. For instance, it is well known that many NP-hard graph problems can be computed efficiently on graphs of bounded \emph{treewidth} using dynamic programming over so-called tree decompositions~\cite{BodlaenderK08}. 
The analysis of computational problems in terms of the input size and an additional parameter such as treewidth is the main objective in the field of parameterized complexity~\cite{DBLP:books/sp/CyganFKLMPPS15,DowneyF13}. 
A parameter similar to treewidth is \emph{treedepth}~\cite[\S 6.4]{NesetrilM12}.
It can be defined as the minimum number of rounds needed to get to the empty graph, where in each round we can delete one vertex from each connected component (formal definitions in the preliminaries). 
Some NP-hard graph problems become solvable in polynomial time if the input graph is restricted to be in a certain class. For instance the NP-hard \textsc{Vertex Cover} can be solved in polynomial time in chordal graphs; those graphs without induced cycles of length at least four.
A parameter that naturally follows from this observation is the minimum cardinality of a set of vertices whose deletion results in a graph contained in graph class $\hh$. Such a set is called an $\hh$-deletion set. This parameter essentially indicates how far the problem is from being a trivial case (cf.~\cite{GuoHN04}). 
The size of a feedback vertex set~\cite{KratschS10,JansenB13} or vertex cover number~\cite{FellowsLMRS08,FluschnikNSZ20} of the graph are often used examples of such parameters, where $\hh$ is the class of forests and edgeless graphs respectively. 

Recently there has been a push~\cite{DBLP:conf/mfcs/EibenGHK19,GanianOS19,DBLP:conf/stacs/GanianRS17} in obtaining parameterized algorithm where the parameter is a hybrid of some overall structure of the graph, like treewidth and treedepth, and some distance to triviality. One such example introduced by Bulian and Dawar is $\hh$-elimination distance ($\hhdepth$)~\cite{DBLP:journals/algorithmica/BulianD16,DBLP:journals/algorithmica/BulianD17}, which can be defined as the minimum number of deletion rounds needed to obtain a graph in $\hh$ by removing one vertex from each connected component in each round; recall that in the elimination-based definition of treedepth, the goal is to eliminate the entire graph. Hence~$\hhdepth$ is never larger than the treedepth or the (vertex-)deletion distance to~$\hh$. Bulian and Dawar showed that~$\hhdepth$ can be computed in FPT time when $\hh$ is minor-closed~\cite{DBLP:journals/algorithmica/BulianD17}.

A related hybrid variant of treewidth was introduced by Eiben et al.~\cite{DBLP:conf/mfcs/EibenGHK19}, namely $\hh$-treewidth ($\hhtw$). The $\hh$-treewidth of a graph can be defined as the minimum treewidth of the \emph{torso graph} of a vertex set whose removal ensures each component belongs to~$\hh$. This gives rise to tree decompositions in which each bag has size at most~$k+1$, apart for an arbitrarily large set of vertices that occurs in no other bags and induces a subgraph from~$\hh$. Similarly as before,~$\hhtw(G)$ is not larger than~$\tw(G)$ or the deletion distance from~$G$ to~$\hh$. For minor-closed graph classes~$\hh$ it can be shown that graphs of~$\hh$-treewidth at most~$k$ are minor-closed and therefore characterized by a finite set of forbidden minors. This leads to non-uniform algorithms to recognize graphs of $\hh$-treewidth at most~$k$ for minor-closed~$\hh$ using the Graph Minor algorithm~\cite{RobertsonS95b}. 
% What is known about computing $\hh$-treewidth?

Apart from minor-closed families~$\hh$, some isolated results are known about FPT algorithms to compute~$\hhdepth$ and~$\hhtw$ exactly, parameterized by the parameter value. In recent work, Agrawal and Ramanujan~\cite{AgrawalR2020} give an FPT algorithm to compute the elimination distance to a cluster graph, as part of a kernelization result using the corresponding structural parameterization. Eiben et al.~\cite{DBLP:conf/mfcs/EibenGHK19} show that when~$\hh$ is the class of graphs of rankwidth at most~$c$ for some constant~$c$, then $\hhtw$ is FPT. Bulian and Dawar~\cite{DBLP:journals/algorithmica/BulianD16} considered the elimination distance to graphs of bounded degree~$d$ and gave an FPT approximation algorithm. Lindermayr et al.~\cite{DBLP:conf/mfcs/LindermayrSV20} showed that the elimination distance of a \emph{planar} graph to a bounded-degree graph can be computed in FPT time. Very recently, Agrawal et al.~\cite{AgrawalKPRS21} obtained non-uniform FPT algorithms for computing the elimination distance to any family~$\hh$ defined by a finite number of forbidden induced subgraphs, thereby settling the case of bounded-degree graphs as well.

\paragraph{Results and techniques}
We show that $\hhtw$ and $\hhdepth$ are non-uniformly fixed parameter tractable parameterized by the solution value when $\hh$ is the class of bipartite graphs. As a side-product of our proof, we show that~$\hhtw$ is non-uniformly FPT when~$\hh$ is defined by a finite number of forbidden induced subgraphs, generalizing the results of Agrawal et al.~\cite{AgrawalKPRS21} for~$\hhdepth$. The non-uniformity of our algorithms stems from the use of a meta-theorem by Lokshtanov et al.~\cite[Theorem 23]{DBLP:conf/icalp/LokshtanovRSZ18} which encapsulates the technique of \emph{recursive understanding}. This theorem essentially states that for any problem expressible in \emph{Counting Monadic Second Order} (CMSO) logic, the effort of classifying whether the problem is in FPT is reduced to inputs that are $(s,c)$-unbreakable (formally defined later). The theorem allows us to use the technique of \emph{recursive understanding} in a black box matter, leading to a streamlined proof at the expense of obtaining non-uniform algorithms. We believe that uniform algorithms can be obtained using the same approach by implementing the recursive understanding step from scratch and deriving an explicit bound on the sizes of representatives for the canonical congruence for~$\hhdepth$ and~$\hhtw$ on $t$-boundaried graphs. As the running times would not be practical in any case, we did not pursue this route.

Our proof is independent of that of Agrawal et al.~\cite{AgrawalKPRS21}, but is based on an older approach inspired by the earlier work of Ganian et al.~\cite{DBLP:conf/stacs/GanianRS17} that contains similar ideas. The key ingredient for our work is the insight that the approach based on recursive understanding used by Ganian et al.~\cite{DBLP:conf/stacs/GanianRS17} to compute a hybrid parameterization for instances of constraint satisfaction problems, can be applied more generally to aid in the computation of~$\hhdepth$ and~$\hhtw$. We can lift one of their main lemmas to a more general setting, where it roughly shows that given a~$(s(k),2k)$-unbreakable graph~$G$ (definitions in Section~\ref{sec:prelims}) and a deletion set~$X$ from~$G$ to~$\hh$ that is a \emph{subset} of some (unknown) structure that witnesses the value of~$\hhtw$ or~$\hhdepth$, we can determine in FPT time whether such a witness exists. This allows~$\hhdepth$ and~$\hhtw$ to be computed in FPT time if we can efficiently find a deletion set with the stated property. For families~$\hh$ defined by finitely many forbidden induced subgraphs, a simple bounded-depth branching algorithm suffices. Our main contribution is for bipartite graphs, where we show that the relation between odd cycle transversals and graph separators that lies at the heart of the iterative compression algorithm for OCT~\cite{DBLP:journals/orl/ReedSV04}, can be combined with the fact that there are only few minimal~$(u,v)$-separators of size at most~$2k$ in~$(s(k),2k)$-unbreakable graphs, to obtain an $\hh$-deletion set with the crucial property described above.

%\bmpr{Mention some consequences of the finite forbidden subgraphs case: say that we solve the open problem of bounded-degree elimination distance, and cite the papers asking that open problem. It would be good to rephrase the relation to previous work. Instead of saying `we do the same as Ganian', we can say something like: ``We extract an ingredient from the work of Ganian et al and show that it can be applied a more general setting. Basically, it tests whether there is a decomposition whose ``witness structure'' (to be defined in Section 2) is a superset of a given $H$-deletion set. Then we develop algorithms to generate an $H$-deletion set with this property, if one exists, in highly connected graphs.'' Then instead of writing ``The finite forbidden subgraph algorithm follows almost directly from their work'' we can write ``For finite forbidden subgraphs, it is easy to find such a deletion set by a bounded-depth branching algorithm.'}

\paragraph{Related work} %Our work is heavily inspired by that of Ganian et al~\cite{DBLP:conf/stacs/GanianRS17}. They show that they can compute a hybrid parameter in FPT time for CSP instances that is similar to $\hhtw$; the remaining components do not belong to some graph class $\hh$, but to some tractable CSP instances. The arguments for the case when $\hh$ is characterized by a finite set of forbidden induced subgraphs follows almost directly from their work. We note that since their paper precedes that of Lokshtanov et al.~\cite{DBLP:conf/icalp/LokshtanovRSZ18}, they do not use the black box theorem we use here. Instead, they put in significant effort to apply the recursive understanding technique directly.

Hols et al.~\cite{HolsKP20} used parameterizations based on elimination distance to obtain kernelization algorithms for \textsc{Vertex Cover}.

%In terms of computing $\hhdepth$, some previous work has been done as well. %For instance, Lindermayr et al.~\cite{DBLP:conf/mfcs/LindermayrSV20} show that when $\hh$ is the set of graphs of degree at most $d$, computing $\hhdepth$ is in FPT when the input graph is planar. Furthermore, Agrawal and Ramanujan~\cite{AgrawalR2020} show that computing $\hhdepth$ is in FPT when $\hh$ is the class of cluster graphs, that is, the class of graphs where each connected component induces a clique. 
In recent work~\cite{JansenKW21}, a superset of the authors gave FPT algorithms to \emph{approximate} $\hhdepth$ and $\hhtw$ for several classes~$\hh$, including bipartite graphs and all classes defined by a finite set of forbidden induced subgraphs. That work employed completely different techniques than used here, and left open the question whether the parameters can be computed exactly in FPT time.

%\paragraph{Some other things to include?}
%\begin{itemize}
%    \item Grohe et al.~\cite{DBLP:conf/stoc/GroheKMW11} Finding topological subgraphs is fixed-parameter tractable. (introduced recursive understanding).
%    \item Chitnis et al.~\cite{DBLP:journals/siamcomp/ChitnisCHPP16} Designing FPT algorithms for cut problems using randomized contractions. Used recursive understanding to solve unique label cover (among others).
%\end{itemize}

\section{Preliminaries}\label{sec:prelims}
We consider simple undirected graphs without self-loops. The vertex and edge set of a graph $G$ are denoted by $V(G)$ and $E(G)$ respectively. When the graph is clear from context, we denote $|V(G)|$ by $n$ and $|E(G)|$ by $m$. For each $X \subseteq V(G)$, the graph induced by $X$ is denoted by $G[X]$. We denote $G[V(G) \setminus X]$ by $G-X$, and write $G-v$ instead of $G- \{v\}$. The open and closed neighborhoods of $v \in V(G)$ are denoted $N_G(v)$ and $N_G[v]$ respectively. For $X \subseteq V(G)$, $N_G[X] = \bigcup_{v \in X} N_G[v]$ and $N_G(X) = N_G[X] \setminus X$. The subscript $G$ is omitted if it is clear from context. The graph obtained from $G$ by contracting an edge $e = \{u,v\} \in E(G)$ is the graph obtained by deleting $u$ and $v$ and inserting a new vertex that is adjacent to all of $(N_G(u) \cup N_G(v)) \setminus \{u,v\}$. A graph $H$ is a \emph{minor} of $G$, if it can be obtained from a subgraph of $G$ by a number of edge contractions. A parameter is a function that assigns an integer to each graph. A parameter $f$ is minor-closed if $f(H) \leq f(G)$ for each minor $H$ of $G$. The connected components of $G$ are denoted by $\cc(G)$. A set $Y \subseteq V(G)$ is an $\hh$-deletion set if $G-Y \in \hh$. A graph class $\hh$ is hereditary if it is closed under vertex deletion, that is, if $G \in \hh$, then for every induced subgraph $F$ of $G$ it holds that $F \in \hh$. In this work we restrict ourselves to hereditary graph classes.
A proper $c$-coloring of a graph is a function $f \colon V(G) \to [c]$ such that for every $\{u,v\} \in E(G)$ it holds that $f(u) \neq f(v)$. A graph is bipartite if and only if it has a proper 2-coloring. For sets $X,Y \subseteq V(G)$, we say that $S \subseteq V(G)$ is an $(X,Y)$-separator if the graph $G-S$ does not contain a vertex $u \in X \setminus S$ and $v \in Y \setminus S$ in the same connected component.

A parameterized problem $\Pi$ is a subset $\Sigma^* \times \mathbb{N}$ for some finite alphabet $\Sigma$. A parameterized problem is \emph{non-uniformly fixed-parameter tractable} (FPT) if there exists a fixed $d$ such that for every fixed $k \in \mathbb{N}$, there exists an algorithm that determines whether $(x,k) \in \Pi$ in $\Oh(|x|^d)$ time. (Hence there is a different algorithm for each value of~$k$.)

\subsection{Treewidth}

\begin{definition}
A \emph{tree decomposition} of a graph $G$ is a pair $(T,\{X_t\}_{t \in V(T)})$, where $T$ is a tree and each $t \in V(T)$ is assigned a vertex subset $X_t \subseteq V(G)$, such that the following holds:
\begin{enumerate}
    \item For every $\{u,v\} \in E(G)$, there exists $t \in V(T)$ with $\{u,v\} \subseteq X_t$.
    \item $\bigcup_{t \in V(T)} X_t = V(G)$.
    \item For every $v \in V(G)$, the set $T_u = \{t \in V(T) \mid u \in X_t\}$ induces a connected subtree of $T$.
\end{enumerate}
The \emph{width} of tree decomposition $(T,\{X_t\}_{t \in V(T)})$ equals $\max_{t \in V(T)} |X_t|-1$. The \emph{treewidth} of a graph, denoted $\tw(G)$, is the minimum possible width over all possible tree decompositions of $G$.
\end{definition}

\subsection{$\hh$-treewidth and $\hh$-elimination distance}

\begin{definition}\cite[Definition 4]{DBLP:conf/stacs/GanianRS17}
Let $G$ be a graph and $X \subseteq V(G)$. The \emph{torso of $X$}, denoted by $\torso_G(X)$, is the graph obtained by turning the neighborhood of every connected component of $G-X$ into a clique, followed by deleting all of $V(G) \setminus X$. 
\end{definition}

Eiben et al.~\cite{DBLP:conf/mfcs/EibenGHK19} use the term of \emph{collapsing} $V(G) \setminus X$ instead of the torso of $X$. Since our algorithms try to identify $X$, the torso terminology is more natural.

\begin{definition}\cite[Definition 3]{DBLP:conf/mfcs/EibenGHK19}
The $\hh$-treewidth of a graph $G$ is the smallest integer $k$ such that there exists a set $X \subseteq V(G)$ with $\tw(\torso_G(X)) \leq k$ and for each connected component $C \in \cc(G-X)$ we have $C \in \hh$. We call $X$ an $\hhtw$ witness of width $k$.
\end{definition}

\begin{definition}\cite{DBLP:journals/algorithmica/BulianD16,DBLP:journals/algorithmica/BulianD17} \label{def:hhdepth}
The $\hh$-elimination distance of $G$ for a hereditary graph class $\hh$, denoted by $\hhdepth(G)$, is defined as:
\begin{align*}
    \hhdepth(G) = 
    \begin{cases}
    \max_{C \in \cc(G)} \hhdepth(C) & \text{ if $G$ not connected} \\
    0 & \text{ if $G$ connected and $G \in \hh$}\\
    1 + \min_{v \in V(G)} \hhdepth(G-v) & \text{ otherwise}
    \end{cases}
\end{align*}
The treedepth of a graph, denoted $\td(G)$, is equivalent to $\hhdepth(G)$ where $\hh$ only contains the empty graph.
\end{definition}

Note that the definition above is well defined when $\hh$ is hereditary, since each hereditary graph class contains the empty graph. We argue that $\hh$-elimination distance has an equivalent definition similar to that of $\hh$-treewidth.

\begin{proposition}\label{prop:H-ed_to_torso-td}
A graph has $\hhdepth(G) \leq k$ if and only if there exists~$X \subseteq V(G)$ such that $\td(\torso_G(X)) \leq k$ and $C \in \hh$ for each $C \in \cc(G-X)$.
\end{proposition}
\begin{proof}
For the first direction, we prove by induction on~$k$ that if~$G$ has a set~$X \subseteq V(G)$ such that~$G-X \in \hh$ and~$\td(\torso_G(X))\leq k$, then~$\hhdepth(G) \leq k$.

For the base case~$k=0$, note that~$\td(\torso_G(X)) = 0$ implies that~$X = \emptyset$, so that~$G \in \hh$. By Definition~\ref{def:hhdepth} we have~$\hhdepth(G) = 0 \leq k$.

For the induction step we have~$k>0$. To show that~$\hhdepth(G) \leq k$, by Definition~\ref{def:hhdepth} it suffices to prove that each~$C \in \cc(G)$ satisfies~$\hhdepth(C) \leq k$. Let~$X_C := C \cap X$. If~$X_C = \emptyset$ then~$C \in \hh$ (since~$\hh$ is hereditary) so~$\hhdepth(C) = 0 \leq k$. In the remainder assume that~$X_C \neq \emptyset$. Observe that~$\torso_G(X)$ has~$\torso_C(X_C)$ as a connected component, and that~$\torso_C(X_C)$ is connected since~$C$ is connected and~$X_C \subseteq C$. 
By definition of~$\td$ there exists a vertex~$x \in X_C$ such that~$\td(\torso_C(X_C) - x) = \td(\torso_C(X_C)) - 1$. Let~$X_C' := X_C \setminus \{x\}$ and let~$C' := C - x$. 
Note that~$\torso_{C'}(X_C') = \torso_C(X_C) - x$: it makes no difference whether we first turn the neighborhood of each component of~$C-X_C$ into a clique, remove~$C \setminus X_C$, and then remove~$x$, or whether we start from~$C' = C - x$, turn the neighborhood of each component of~$C' - X_C'$ into a clique, and then remove~$C' \setminus X_C'$. 
By induction on~$C'$ and~$X_C'$ with~$k' := \td(\torso_{C'}(X_C')) < k$, it follows that~$\hhdepth(C') \leq \td(\torso_{C'}(X_C')) = \td(\torso_C(X_C)) - 1 \leq \td(\torso_G(X)) - 1$, where the last inequality follows since~$\torso_C(X_C)$ is a connected component of~$\torso_G(X)$. By Definition~\ref{def:hhdepth}, since~$C$ is connected we  have~$\hhdepth(C) \leq 1 + \min_{v \in C} \hhdepth(C - v) \leq 1 + \hhdepth(C - x) \leq 1 + (\td(\torso_G(X)) - 1) = \td(\torso_G(X)) \leq k$, which completes this direction of the proof.

For the converse direction, we prove that if~$\hhdepth(G) \leq k$ then~$G$ has a vertex set~$X$ such that~$G - X \in \hh$ and~$\td(\torso_G(X)) \leq k$. We use an induction on~$k + |V(G)|$. If~$\hhdepth(G) = 0$ then by Definition~\ref{def:hhdepth} we have~$G \in \hh$ so that~$X = \emptyset$ suffices. For the induction step we have~$\hhdepth(G) > 0$. We distinguish two cases, depending on the connectivity of~$G$.

If~$G$ is connected, then since~$\hhdepth(G) > 0$ we have~$G \notin \hh$. Hence by Definition~\ref{def:hhdepth} we have~$\hhdepth(G) = 1 + \min_{v \in V(G)} \hhdepth(G - v)$. Let~$x$ be a vertex for which equality is attained. Since~$\hhdepth(G - x) = \hhdepth(G) - 1 < k$, by induction on~$G' := G - x$ there exists a set~$X' \subseteq V(G')$ such that~$\td(\torso_{G'}(X')) \leq k-1$. Define~$X := X' \cup \{x\}$ and note that~$\td(\torso_G(X)) \leq 1 + \td(\torso_{G'}(X')) \leq 1 + (k-1)$ since the graph~$\torso_{G'}(X')$ can be obtained from~$\torso_G(X)$ by removing the vertex~$x$. Hence~$\td(\torso_G(X)) \leq k$, proving the claim.

Now suppose that~$G$ is disconnected, so that~$\hhdepth(G) = \max_{C \in \cc(G)} \hhdepth(C)$. For each~$C \in \cc(G)$ we have that~$|V(C)| < |V(G)|$ and~$\hhdepth(C) \leq \hhdepth(G) = k$, so we may apply the induction hypothesis to~$C$ to obtain a set~$X_C \subseteq V(C)$ such that~$\td(\torso_C(X_C)) \leq \hhdepth(C) \leq k$. Let~$X := \bigcup_{C \in \cc(G)} X_C$. Observe that each connected component of the graph~$\torso_G(X)$ is equal to~$\torso_C(X_C)$ for some~$C \in \cc(G)$, so that each connected component~$H$ of~$\torso_G(X)$ satisfies~$\td(H) \leq \td(\torso_C(X_C)) \leq k$ for some~$C \in \cc(G)$. By Definition~\ref{def:hhdepth}, the fact that each component of~$\torso_G(X)$ has treedepth at most~$k$ ensures~$\td(\torso_G(X)) \leq k$, which concludes the proof.
\end{proof}

Similar to $\hhtw$ witnesses, we call $X$ an $\hhdepth$ witness of depth $k$. Since the torso operation on $X$ turns the neighborhood of each connected component of $G-X$ into a clique, the following note follows.

\begin{note}\label{note:smallneighborhood}
If $X$ is a $\hhtw$ witness of width $k-1$ (respectively $\hhdepth$ witness of depth $k$), then $|N(C)| \leq k$ for every $C \in \cc(G-X)$.
\end{note}

We are ready to introduce the main problem we try to solve.

\defparproblem{$\hh$-treewidth ($\hhtw$) / $\hh$-elimination distance ($\hhdepth$)}{A graph $G$, an integer $k$.}{$k$}{Decide whether $\hhtw(G) \leq k-1$ / $\hhdepth(G) \leq k$.}{Question}

\begin{definition}\cite{DBLP:conf/icalp/LokshtanovRSZ18}
Let $G$ be a graph and $s, c \in \mathbb{N}$. A partition $(X,C,Y)$ of $V(G)$ is an $(s,c)$-separation in $G$ if:
\begin{itemize}
    \item $C$ is a separator, that is, no edge has one endpoint in $X$ and one in $Y$,
    \item $|C| \leq c$, $|X| \geq s$, and $|Y| \geq s$.
\end{itemize}
A graph $G$ is $(s,c)$-unbreakable if there is no $(s,c)$-separation in $G$.
\end{definition}

The following proposition is similar to Lemma 21 of Ganian et al.~\cite{DBLP:conf/stacs/GanianRS17}.
\begin{proposition}\label{prop:nice}
Let $G$ be an $(s,c)$-unbreakable graph for $s,c \in \mathbb{N}$ and $\hh$ be a graph class such that $\hhtw(G) \leq k-1$ (resp. $\hhdepth(G) \leq k$) and $c \geq k$. Then at least one of the following holds:
\begin{enumerate}
    \item\label{item:twbounded} $\tw(G) \leq s + k - 1$ (resp. $\td(G) \leq s + k - 1$),
    \item\label{item:onelargecomponent} each $\hhtw$ (resp. $\hhdepth$) witness $X$ of $G$ satisfies the following:
    \begin{itemize}
        \item $G-X$ has exactly one connected component $C$ of size at least $s$, and
        \item $|V(G) \setminus N[C]| < s$ and $|X| \leq s+k-1$
    \end{itemize}
\end{enumerate}
\end{proposition}
\begin{proof}
Consider an arbitrary witness $X$. If all connected components of $G-X$ have size at most $s-1$, then~\ref{item:twbounded} holds. Otherwise, let $C$ be some component of $G-X$ of size at least $s$. First observe that $N(C) \subseteq X$ and $|N(C)| \leq k$ by Note~\ref{note:smallneighborhood}. If $|V(G) \setminus N[C]| \geq s$, then $(C,N(C),V(G) \setminus N[C])$ is an $(s,c)$-separation as $k \leq c$. Since $G$ is $(s,c)$-unbreakable we must have that $|V(G) \setminus N[C]| < s$. Since for any connected component $C'$ of $G-X$ besides $C$ it holds that $V(C') \subseteq V(G) \setminus N[C]$, we get $|V(C')| < s$ too. Finally note that $X \subseteq V(G) \setminus N(C)$ and hence $|X| \leq s + k - 1$ for any witness $X$ and hence~\ref{item:onelargecomponent} holds.
\end{proof}

%\jjh{
%Alternative presentation. Can we solve the problem having the following instead of Proposition~\ref{prop:nice}.

%\begin{lemma}
%Let $G$ be an $(s,c)$-unbreakable graph. Let $X$ be a $\hhtw$ witness of width $k-1$ (respectively $\hhdepth$ witness of depth $k$) of $G$, where $c \geq k$, then $|X| \leq 3(s(k)+k)$.
%\end{lemma}
%\begin{proof}
%Suppose that $|X| > 3(s(k)+k)$. By Lemma 7.20 of~\cite{DBLP:books/sp/CyganFKLMPPS15}, it follows that $\torso_G(X)$ has a ...
%\end{proof}
%}

%\jjh{give specific bound based on whiteboard discussion}

The following lemma bounds the number of small connected vertex sets with a small neighborhood. It was originally stated for connected sets of exactly~$b$ vertices with an open neighborhood of exactly~$f$ vertices.

\begin{lemma}\cite[cf.~Lemma 3.1]{DBLP:journals/combinatorica/FominV12}\label{lem:enumerate_small_connected_graphs}
Let $G$ be a graph. For every $v \in V(G)$ and $b,f \geq 0$, the number of connected vertex sets $B \subseteq V(G)$ such that (a) $v \in B$, (b) $|B| \leq b+1$, and (c) $|N(B)| \leq f$ is at most $b \cdot f \cdot \binom{b+f}{b}$. Furthermore they can be enumerated in $\Oh(n \cdot b^2 \cdot f \cdot (b+f) \cdot \binom{b+f}{b})$ time using polynomial space.
\end{lemma}

%We note that the lemma above was originally stated for connected vertex sets $B$ with $|B|=b+1$ and $|N(B)| = f$, but the version above easily follows from their proof and has been used before in the literature (see~\cite{DBLP:conf/soda/FominLMRS15} for instance).

%\jjh{Two other occurences are full version of "reducing cmso to highly connected graphs" lemma 4.2, and in proof of lemma 9 of "combining tw and backdoors", but may not be worth mentioning. many papers citing~\cite{DBLP:journals/combinatorica/FominV12} are doing it not for this lemma, so not easy to find. }

\subsection{CMSO}
Our exposition of CMSO roughly follows that of Lokshtanov et al.~\cite{DBLP:conf/icalp/LokshtanovRSZ18}. Monadic second order logic (MSO) is a logic that can be used to express properties of graphs. The syntax includes logical connectives such as $\vee$, $\wedge$, $\neg$, $\Leftrightarrow$, $\Rightarrow$, and variables for single vertices, single edges, sets of vertices, and sets of edges, which can be quantified using $\forall$ and $\exists$. Furthermore there are binary relations for set membership ($\in)$, equality of variables ($=$), testing whether edge $e$ incident to vertex $v$ ($\inc(v,e)$), and finally testing whether two vertices are adjacent ($\adj(u,v)$).
Counting monadic second order logic (CMSO) is an extension of MSO that includes a cardinality test $\textbf{card}_{q,r}(S)$, which is true if and only if $|S| \equiv q \mod r$.
For a more complete introduction to CMSO we refer to the book of Courcelle and Engelfriet~\cite{CourcelleEngelfriet2012}.

Let $\hh$ be a graph class. We say that containment in $\hh$ is expressible in CMSO if there exists a CMSO formula $\varphi_\hh$ such that for any graph $G$ it holds that $G \models \varphi_\hh$ if and only if $G \in \hh$.

\begin{lemma}\label{lem:cmsoformulas}
There exist CMSO-formulas with the following properties:
\begin{enumerate}
    \item\label{item:minor} For any graph $H$, there exists a formula~$\varphi_{\mathrm{H-MINOR}}(X)$ such that for any graph $G$ and any $X \subseteq V(G)$ it holds that $(G,X) \models \varphi_{\mathrm{H-MINOR}}(X)$ if and only if $H$ is a minor of $G[X]$.
    \item\label{item:forbiddensubgraph} For any graph class $\hh$ characterized by a finite set of forbidden induced subgraphs, there exists a formula~$\varphi_{\hh}$ such that for any graph~$G$ it holds that~$G \models \varphi_{\hh}$ if and only if graph~$G \in \hh$.
    \item\label{item:bipartite} There exists a formula~$\varphi_{BIP}$ such that for any graph~$G$ it holds that~$G \models \varphi_{BIP}$ if and only if graph~$G$ is bipartite.
    \item\label{item:final} For each~$k \in \mathbb{N}$, for each graph class $\hh$ such that containment in $\hh$ is CMSO expressible, and for each minor-closed parameter $f$, there exists a formula~$\varphi_{(k,\hh,f)}(X)$ such that for any graph~$G$ and any~$X \subseteq V(G)$ we have~$(G,X) \models \varphi_{(k,\hh,f)}(X)$ if and only if~$f(\torso_G(X)) \leq k$ and $C \in \hh$ for each $C \in \cc(G-X)$.
\end{enumerate}
\end{lemma}
\begin{proof}
For~\ref{item:minor} see for instance Corollary 1.14~\cite{CourcelleEngelfriet2012}, we repeat it here as we adapt it for~\ref{item:final}.
\begin{align*}
    \textsc{conn}(X,V,E) = \forall Y \subseteq V (( \exists u \in X : u \in Y \wedge \exists v\in X : v \notin Y)  \Rightarrow\\
    (\exists e \in E \exists u,v \in X : \inc(u,e) \wedge \inc(v,e) \wedge u \in Y \wedge v \notin Y))\\
   \varphi_{\mathrm{H-MINOR}}(X) = \exists Y_1,\dots,Y_n \subseteq X : (\bigwedge_{1 \leq i \leq n} ((\exists y : y \in Y_i) \wedge \textsc{conn}(Y_i, V, E)) \\
    \wedge \bigwedge_{1 \leq i < j \leq n} \neg \exists y ( y\in Y_i \wedge y \in Y_j) 
    \wedge \bigwedge_{(i,j) \in E(H)} \exists u,v (u \in Y_i \wedge v \in Y_j \wedge \adj(u,v)))
\end{align*}

For~\ref{item:forbiddensubgraph}, let $\mathcal{F}_\hh$ be the forbidden induced subgraph characterization of $\hh$, where $H \in \mathcal{F}_\hh$ is a graph on vertex set $[|V(H)|]$. A formula $\varphi_\hh$ is given below and is similar to that of checking for a minor.

\begin{align*}
   \varphi_{\hh}= \bigwedge_{H \in \mathcal{F}_\hh} \neg\exists v_1,\dots,v_{|V(H)|} \in V(G) : (\bigwedge_{1 \leq i < j \leq |V(H)|} v_i \neq v_j \\
    \wedge \bigwedge_{(i,j) \in E(H)} \adj(v_i,v_j) \wedge \bigwedge_{(i,j) \notin E(H)} \neg\adj(v_i,v_j))
\end{align*}

Since a graph is bipartite if and only if it has a proper 2-coloring, the following formula shows item~\ref{item:bipartite}. 
\begin{align*}
    \textsc{partition}(V,X_1,X_2) = & \forall_{v \in V}[(v \in X_1 \wedge v \notin X_2) \vee (v \notin X_1 \wedge v \in X_2)] \\
    \textsc{indp}(X) = & \forall_{u,v \in X}\neg \adj(u,v) \\
    \varphi_{BIP} = & \exists_{X_1,X_2 \subseteq V(G)} \textsc{partition}(V(G),X_1,X_2) \wedge \textsc{indp}(X_1) \wedge \textsc{indp}(X_2)
\end{align*}

Finally for~\ref{item:final} note that since $f$ is minor-closed, the set of graphs $F$ with $f(F) \leq k$ has a finite set of forbidden minors by the Graph Minor Theorem of Robertson and Seymour. Using formula~\ref{item:minor}, we can check whether $\torso_G(X)$ contains a forbidden minor. The only thing we need to change is that an edge $\{u,v\}$ is in $\torso_G(X)$ if either $u$ and $v$ are adjacent, or if there is a path whose internal vertices are not in $X$.

\begin{align*}
    \textsc{Tadj}(u,v,X) = \adj(u,v) \vee \exists P \subseteq V(G) (u,v \in P \wedge \textsc{conn}(P,V,E) \\ \wedge \forall w \in P (w = u \vee w = v \vee w \notin X))
\end{align*}

Finally we can check if each connected component $C$ of $\cc(G-X)$ is in $\hh$ by going over every vertex subset and verifying that if it is connected, disjoint from $X$, and maximal, then it induces a graph in $\hh$.
\end{proof}

Since both treewidth and treedepth are minor-closed parameters, we note the following from the lemma above.

\begin{note}\label{note:cmso_express_HTW_HED}
For each $k \in \mathbb{N}$ and graph class $\hh$ such that containment in $\hh$ is CMSO-expressible, there exists a formula $\varphi_{(k,\hh,\tw)}$ (respectively $\varphi_{(k,\hh,\td)}$) such that $(G,k)$ is a \textsc{yes}-instance of $\hh$-\textsc{treewidth} (respectively $\hh$-\textsc{elimination distance}) if and only if $G \models \varphi_{(k,\hh,\tw)}$ (respectively $G \models \varphi_{(k,\hh,\td)}$). 
\end{note}

CMSO formulas can have free variables. A graph together with an evaluation of free variables is called a \emph{structure}. We denote the problem of evaluating a CMSO formula $\varphi$ on a structure by $\textsc{CMSO}[\varphi]$. The following theorem is the main tool used to achieve our algorithms, we apply it only to formulas without free variables. The formulation is slightly different from its original form, see the appendix for details.

\begin{theorem}\cite[Theorem 23]{DBLP:conf/icalp/LokshtanovRSZ18}\label{thm:blackbox}
Let $\hat{\varphi}$ be a CMSO formula. For all $\hat{c} \colon \mathbb{N}_0 \to \mathbb{N}_0$, there exists $\hat{s} \colon \mathbb{N}_0 \to \mathbb{N}_0$ such that if $\textsc{CMSO}[\hat{\varphi}]$ parameterized by $k$ is FPT on $(\hat{s}(k),\hat{c}(k))$-unbreakable structures, then \textsc{CMSO}$[\hat{\varphi}]$ parameterized by $k$ is FPT on general structures.
\end{theorem}

%The length of a formula $\varphi$ is denoted by $||\varphi||$. The following well known Courcelle's theorem states that on bounded treewidth graphs we can evaluate CMSO formulas in linear time.
%\begin{theorem}\cite{DBLP:journals/iandc/Courcelle90}\label{thm:courcelle_FPT}
%Let $\varphi$ be a CMSO formula and $G$ be an $n$-vertex graph together with an evaluation with all free variables of $\varphi$. Suppose a tree decomposition of width $t$ of $G$ is given. Then there exists an algorithm that verifies whether $\varphi$ is satisfied by $G$ in time $f(||\varphi||,t) \cdot n$, for some computable function $f$.
%\end{theorem}

\section{Algorithms for computing $\hhdepth$ and $\hhtw$}
In this section we present our algorithms. In Section~\ref{subsec:witness} we present a key lemma. In Section~\ref{subsec:finite} we use it to deal with $\hh$ characterized by a finite number of forbidden induced subgraphs, and in Section~\ref{subsec:bip} we deal with bipartite graphs.

\subsection{Extracting witnesses from deletion sets contained in them} \label{subsec:witness}

%\defparproblem{$\hh$-deletion enumeration ($k$)}{An $(s(k),c(k))$-unbreakable graph $G$ and an integer $k$.}{$k$}{Return a list $\mathcal{Y}$ of $\hh$-deletion sets of size at most $s(k)+k-1$ with the guarantee that if $\hhtw(G) \leq k-1$ (respectively $\hhdepth(G) \leq k$) and $\tw(G) > s(k) + k$, then there exists $Y \in \mathcal{Y}$ and an $\hhtw$ witness $X$ of $G$ of width $k-1$ (respectively $\hhdepth$ witness $X$ of depth $k$) such that $Y \subseteq X$.}{Task}

Our strategy for solving $\hh$-\textsc{treewidth} and $\hh$-\textsc{elimination distance} is similar to that of lemmas 9 and 10 of Ganian et al.~\cite{DBLP:conf/stacs/GanianRS17} and is based on Proposition~\ref{prop:nice}. Given an $(s(k),c(k))$-unbreakable graph, either the treewidth of the graph is bounded (\ref{item:twbounded}) and we can solve the problem directly using Courcelle's Theorem, or each witness is of bounded size and introduces some structure (\ref{item:onelargecomponent}).

In the following lemma we assume we are in the latter case (hence the $\tw(G) > s(k)+k$ condition) and are given some $\hh$-deletion set $Y$. We show that given an $(s(k),c(k))$-unbreakable graph, in FPT time we can find a witness $X$ such that $Y \subseteq X$ if such a witness exists.

%\jjh{formally define the problem we are solving recursively (i.e. find witness $X$ that is a superset of some set $Z$, where we guess $B \subseteq Y_i$ to be in $Z$ and branch to put more vertices in set along the way), that makes it easier to say something about the running time.}

%\jjh{instead of enumeration, build an algorithm that given an $\hh$-deletion set, decides whether there is a witness that is a superset of this set. for finite forbidden subgraphs, just iterate over all minimal solutions, for oct, construct such a set more carefully...}

\begin{lemma}\label{lem:enumeration_find_witness}
Consider some $k \in \mathbb{N}$ and $c \colon \mathbb{N} \to \mathbb{N}$ such that $c(k) \geq k$. Let $\hh$ be a graph class such that containment in $\hh$ is solvable in polynomial time. There is an algorithm that runs in FPT time that, given an $(s(k),c(k))$-unbreakable graph for any $s \colon \mathbb{N} \to \mathbb{N}$ with $\tw(G) > s(k)+k$ and an $\hh$-deletion set $Y$ of size at most $s(k)+k$, decides whether there is an $\hhtw(G)$ witness $X$ of width at most $k-1$ (respectively $\hhdepth(G)$ witness $X$ of depth at most $k$) such that $Y \subseteq X$.
\end{lemma}
\begin{proof}
We refer to a witness as either being an $\hhtw$ witness of width at most $k-1$ or an $\hhdepth$ witness of depth at most $k$. Given a set $X \subseteq V(G)$, we can verify that it is a witness by testing whether $\tw(\torso_G(X)) \leq k-1$ (respectively~$\td(\torso_G(X)) \leq k$) in FPT time~\cite{Bodlaender96,ReidlRVS14} and verifying that each connected component $C \in \cc(G-X)$ is contained in $\hh$, which can be done in polynomial time by assumption.

We show that we can find a witness if it exists, by doing the above verification for FPT many vertex subsets $D \subseteq V(G)$, as follows.
\begin{enumerate}
    \item For each $y \in Y$, let $\mathcal{C}_y$ be the set of connected vertex sets $S$ with $y \in S$, $|S| \leq s(k)$ and $|N(S)| \leq k$. For each $B \subseteq Y$ with $|B| \leq k$, a choice tuple $t_B$ contains an entry for each $y \in Y \setminus B$, where entry $t_B[y]$ is some set $C_y \in \mathcal{C}_{y}$.
    %\bmpr{'whose size is at most' is a bit ambiguous; in old graph theory papers, the size of a graph is the number of edges and the 'order' of the graph is the number of vertices. I recommend rephrasing to make it explicit what you mean.}
    \item For each $B \subseteq Y$ with $|B| \leq k$ and each choice tuple $t_B$, if $G - (Y \cup \bigcup_{y \in Y \setminus B} N(t_B[y]))$ has one connected component $C$ of size at least $s(k)$ and $|V(G) \setminus N[C]| < s(k)$, apply the witness verification test to $D = Y \cup \bigcup_{y \in Y \setminus B} N(t_B[y]) \cup Q$ for each $Q \subseteq V(G) \setminus N[C]$.
    \item Return the logical or of all witness verification tests.
\end{enumerate}
We argue that the algorithm runs in FPT time. Note that as $|Y| \leq s(k) + k$, there are at most $\binom{s(k)+k}{k}$ choices for $B$. Furthermore $\mathcal{C}_y$ can be computed in FPT time using Lemma~\ref{lem:enumerate_small_connected_graphs}, hence the number of choice tuples is also FPT many. For each choice for $B$ and each choice tuple $t_B$, there are at most $2^{s(k)}$ choices for $Q$. Since each vertex set can be verified to be a witness in FPT time, the running time claim follows.

Finally we argue correctness of the algorithm. Since $\tw(G) > s(k) + k$ (and also $\td(G) > s(k) + k$ as $\tw(G) \leq \td(G) - 1$), by Proposition~\ref{prop:nice} any witness $X$ is of size at most $s(k) + k - 1$, the graph $G-X$ has exactly one large connected component $C$ of size at least $s(k)$, and $|V(G) \setminus N[C]| < s(k)$.

Suppose $G$ has a witness that is a superset of $Y$. Fix some witness $X$ of minimal cardinality with $Y \subseteq X$ and let $C$ be the unique component of size at least $s(k)$ of $G-X$. Note that since $C \cap X = \emptyset$, we have $C \cap Y = \emptyset$. 
%\bmp{Reorder $Y \cap C$ into $C \cap Y$ to be symmetric with $C \cap X$ earlier in sentence.}
%\bmpr{I'm not sure $C$ is well-defined here. By Proposition 2, any witness has a single large component, but a priori different witnesses can have different vertex sets for the single large component. Do you want to fix an arbitrary $C$? Or do you want to say: let $X$ be a minimum-size witness that is a superset of $Y$ and define $C$ to be its large leftover component?}

Let $B = N(C) \cap Y$. By Note~\ref{note:smallneighborhood} we have $|N(C)| \leq k$, hence the branching algorithm makes this choice for $B$ at some point. For each $y \in Y \setminus B$, let $C_y$ be the connected component of $G-N[C]$ containing $y$. \jjh{See Figure~\ref{fig:find_witness} for a sketch of the situation.} Since $|V(G) \setminus N[C]| < s(k)$ and $|N(C)| \leq k$, we have that $|V(C_y)| < s(k)$ and $|N(C_y)| \leq k$. Note that $N(C_y) \subseteq N(C) \subseteq X$.
%\bmpr{Add: $\subseteq X$, so make it easier to see that $A\ \subseteq X$?}
The branching algorithm at some point tries the choice tuple $t_B$ where $t_B[y] = C_y$ for each $y \in Y \setminus B$.
Consider the set $A = Y \cup \bigcup_{y \in Y \setminus B} N(t_B[y])$. Note that $A \subseteq X$ by construction. 

If $N(C) \subseteq A$, then the single large component of $G-A$ of size at least $s(k)$ is exactly $C$. Since $|V(G) \setminus N[C]| < s(k)$, it follows that $X = A \cup Q$ for some $Q \subseteq V(G) \setminus N[C]$. It follows that the algorithm correctly identifies $X$ in this case. 

The only remaining case is $N(C) \not\subseteq A$. We argue that this cannot happen when witness $X$ is of minimal cardinality. Suppose~$N(C) \not \subseteq A$ and let $v \in N(C) \setminus A$. Let $Z = Y \cup \bigcup_{y \in Y \setminus B} N[C_y]$ and note that we subtract the \emph{closed} neighborhoods of the components, instead of the \emph{open} neighborhoods as in the definition of~$A$. 
Let $C_v^*$ be the connected component of $G-(C \cup Z)$ that contains $v$. We argue that $X \setminus C_v^*$ is a witness. \jjh{Again consult Figure~\ref{fig:find_witness} for an intuition.}
%\bmpr{The 'clearly' here could use some explanation or reference to a figure, since you are comparing components of$G - (C \cup A)$ to components of $G - (X \setminus C_v)$.} 
Note that $C_v^* \cap Y = \emptyset$ by construction as $Y \subseteq Z$. Because $Y$ is an $\hh$-deletion set, it follows that for each connected component $C'$ in $G - (X \setminus C_v^*)$ we have $C' \in \hh$. 
%We argue that $N(C_v^*) \cap X \subseteq N(C)$.\jjh{\{actually $N(C_v^*) \setminus C \subseteq N(C)$, but does not matter for correctness so keep like it is?\}}
\jjh{We argue that $N(C_v^*) \subseteq N[C]$. First we argue that $N(C_v^*) \cap X \subseteq N(C)$}.
Indeed if any vertex in $X \setminus (Z \cup N(C))$ was adjacent to $C_v^*$, the vertex itself would belong to $C_v^*$. If any vertex $z \in Z \setminus N(C)$ was adjacent to $w \in C_v^*$, then either $w \in N(C)$ and hence $w \in N(C_y)$ for some $y \in Y$ and hence $w \in Z$, or $w \in C_v^* \setminus N(C)$ and $w \in C_y$ for some $y \in Y$ and hence $w \in Z$; in both cases we contradict $w \in C_v^*$. \jjh{Similar arguments show that $N(C_v^*) \setminus X \subseteq C$. Since $N(C_v^*) \subseteq N[C]$ and $v$ is adjacent to at least one vertex in $C$ as $v \in N(C)$, it follows that $C \cup C_v^*$ is a connected component of $G - (X \setminus C_v^*)$ with $N(C \cup C_v^*) \subseteq N(C)$.}
Therefore $\torso_G(X \setminus C_v^*)$ is an induced subgraph of $\torso_G(X)$. We conclude that $X \setminus C_v^*$ is a witness. Since $X$ was assumed to be of minimal cardinality, we arrive at a contradiction and hence $A \supseteq N(C)$.

\begin{figure}
    \centering
    \includegraphics[page=1]{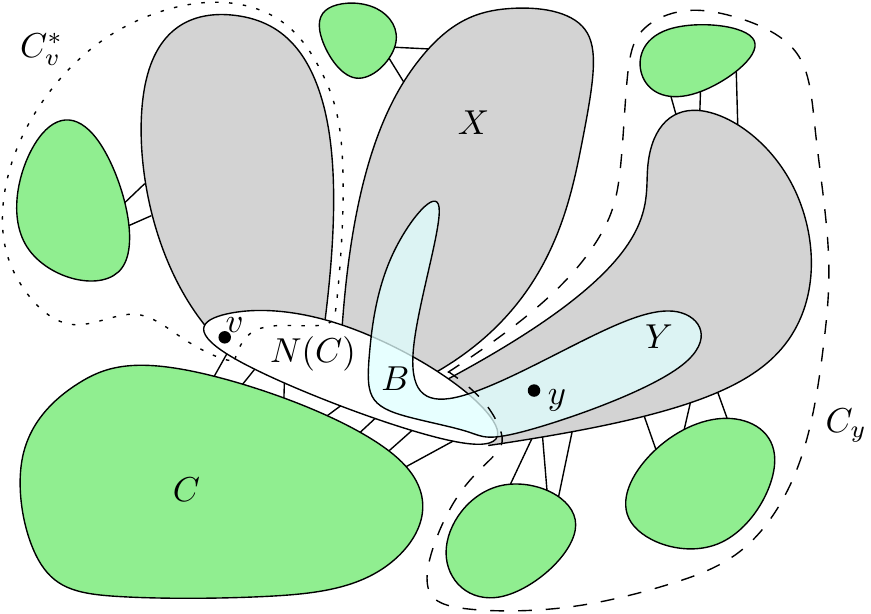}
    \caption{\jjh{Situation sketch of Lemma~\ref{lem:enumeration_find_witness}. The set $X$ in grey denotes a witness and the set $C$ is the single large component of $G-X$.}}
    \label{fig:find_witness}
\end{figure}

\end{proof}

\subsection{Classes $\hh$ with finitely many forbidden induced subgraphs} \label{subsec:finite}

\begin{theorem} \label{thm:finiteobstructions:fpt}
Let $\hh$ be a graph class characterized by a finite set of forbidden induced subgraphs. Then $\hh$-\textsc{treewidth} and $\hh$-\textsc{elimination distance} are non-uniformly fixed-parameter tractable.
\end{theorem}
\begin{proof}
By Lemma~\ref{lem:cmsoformulas} containment in $\hh$ is CMSO expressible, therefore by Note~\ref{note:cmso_express_HTW_HED} there exists a formula $\varphi_{(k,\hh,f)}$ for each $f \in \{\tw,\td\}$ such that an instance $(G,k)$ of $\hh$-\textsc{treewidth} (respectively $\hh$-\textsc{elimination distance}) is a \textsc{yes}-instance if and only if $G \models \varphi_{(k,\hh,f)}$. Furthermore, containment in $\hh$ is polynomial time solvable, as we can verify that a graph does not contain any of the finitely many forbidden induced subgraphs.

We argue that both problems are in FPT when the input graph $G$ is $(s(k),k)$-unbreakable for any $s \colon \mathbb{N} \to \mathbb{N}$. If $\tw(G) \leq s(k) + k$, we solve the problems directly using Courcelle's Theorem~\cite{DBLP:journals/iandc/Courcelle90} using $\varphi_{(k,\hh,f)}$. Otherwise by Proposition~\ref{prop:nice} each witness $X$ is of size at most $s(k)+k-1$. We can enumerate all minimal $\hh$-deletion sets $\mathcal{Y}$ of size at most $s(k)+k-1$ in FPT time by finding a forbidden induced subgraph and branching in all finitely many ways of destroying it. Since any witness $X$ is an $\hh$-deletion set, for some $Y \in \mathcal{Y}$ we have $Y \subseteq X$. Hence we solve the problem by calling Lemma~\ref{lem:enumeration_find_witness} for each $Y \in \mathcal{Y}$. Applying Theorem~\ref{thm:blackbox} concludes the proof.
\end{proof}

Using known characterizations by a finite number of forbidden induced subgraphs (cf.~\cite{BrandstadtLS99}) we obtain the following corollary to Theorem~\ref{thm:finiteobstructions:fpt}.

\begin{corollary}
Let $\hh$ be set of graphs that are either (1)~cliques, (2)~claw-free, (3)~of degree at most~$d$ for fixed~$d$, (4)~cographs, or~(5) split graphs. $\hh$-\textsc{treewidth} and $\hh$-\textsc{elimination distance} are non-uniformly fixed-parameter tractable.
\end{corollary}

% Note: commented out because we do not need to solve the bounded-degree case anymore.
%
%\bmp{For general families~$\hh$ defined by forbidden induced subgraphs on at most~$d$ vertices, the running time is of the form~$f(k) \cdot n^{d + \Oh(1)}$ due to the necessity of testing for the existence of a forbidden induced subgraph in time~$\Omega(n^d)$ during the enumeration process. However, when~$\hh$ is the class of bounded-degree graphs this can be avoided since an obstruction is simply formed by a vertex of degree larger than~$d$, which can be found in linear time. Hence the running time for testing whether the elimination distance to a graph of degree at most~$d$ is bounded by~$k$, is of the form~$f(k,d) \cdot n^{\Oh(1)}$.}

\subsection{Bipartite graphs} \label{subsec:bip}
We use shorthand $\bip$ to denote the class of bipartite graphs. The problem of deleting $k$ vertices to obtain a bipartite graph is better known as the \textsc{Odd Cycle Transversal (OCT)} problem. The problem was shown to be FPT for the first time by Reed et al.~\cite{DBLP:journals/orl/ReedSV04}. We use some of their ingredients to show the following.

%\defproblem{Annotated bipartite coloring (abc)}{A bipartite graph $G$, two sets $B_1,B_2 \subseteq V(G)$, and an integer $k$.}{Return a minimum-cardinality set $X \subseteq V(G)$ such that $G-X$ has a proper 2-coloring $f$ with $f(v) = i$ for each $v \in B_i \setminus X$, or return $\bot$ if no such $X$ exists of size at most $k$.}

%Consider some instance $(G,B_1,B_2,k)$ of \textsc{abc}. Let $f^*$ be an arbitrary proper 2-coloring of $G$ and let $B_i^* = (f^*)^{-1}(i)$ for each $i \in [2]$. Let $A := (B_1 \cap B_2^*) \cup (B_2 \cap B_1^*)$ and $R := (B_1 \cap B_1^*) \cup (B_2 \cap B_2^*)$. It turns out that any solution must separate $A$ from $R$.

%\begin{lemma}\cite[Lemma 4.15]{DBLP:books/sp/CyganFKLMPPS15}\label{lemma:CRsep}
%A set $X$ is a solution for \textsc{abc} instance $(G,B_1,B_2,k)$ if and only if $X$ separates $A$ and $R$, that is, no component of $G-X$ contains vertices from both $A \setminus X$ and $R \setminus X$. Furthermore, such a set $X$ of size $k$ (provided it exists) can be found in time $\Oh(k(n+m))$.
%\end{lemma}

%With the above lemma at hand, we show the main lemma of this section. Let $\bip$ be the class of bipartite graphs.
%\bmpr{Define $\bip$}
\begin{lemma}\label{lem:bip_enum}
The $\bip$-\textsc{treewidth} and $\bip$-\textsc{elimination distance} problems are non-uniformly fixed-parameter tractable.
\end{lemma}
\begin{proof}
By Lemma~\ref{lem:cmsoformulas} containment in the class of bipartite graphs is CMSO expressible, therefore by Note~\ref{note:cmso_express_HTW_HED} there exists a formula $\varphi_{(k,\bip,f)}$ for each $f \in \{\tw,\td\}$ such that an instance $(G,k)$ of $\bip$-\textsc{treewidth} (respectively $\bip$-\textsc{elimination distance}) is a \textsc{yes}-instance if and only if $G \models \varphi_{(k,\bip,f)}$. We argue that both problems are FPT in $(s(k),2k)$-unbreakable graphs for any $s \colon \mathbb{N} \to \mathbb{N}$. Note that the theorem then follows by Theorem~\ref{thm:blackbox}.

Let~$G$ be an $(s(k),2k)$-unbreakable graph. As before, we use the term witness to either refer to an $\hhtw$ witness of width at most $k-1$ or an $\hhdepth$ witness of depth at most $k$, depending on the problem being solved. We first test whether $\tw(G) \leq s(k)+k$, in FPT time~\cite{Bodlaender96}. If so, then we can solve the problems directly using Courcelle's Theorem~\cite{DBLP:journals/iandc/Courcelle90} using $\varphi_{(k,\bip,f)}$. Otherwise by Proposition~\ref{prop:nice} the size of each witness in~$G$ is at most $s(k)+k-1$, and for each witness~$X$ there is a unique connected component of~$G-X$ of at least~$s(k)$ vertices, henceforth called the \emph{large component}. We use a two-step process to find an odd cycle transversal that is a subset of some witness (if a witness exists), so that we may invoke Lemma~\ref{lem:enumeration_find_witness} to find a witness.

For a witness~$X^*$ in~$G$ and an odd cycle transversal~$W$ of~$G$, we say that a partition~$(W_L, W_I)$ of~$W$ is \emph{weakly consistent} with~$X^*$ if for the unique large component~$C$ of~$G - X^*$ we have that~$W \cap C = W_L$, $|W_L| \leq k$, and~$W \subseteq C \cup X^*$. An odd cycle transversal~$W$ is \emph{strongly consistent} with~$X^*$ if~$W \subseteq X^*$.

The following claim encapsulates the connection between odd cycle transversals and separators that forms the key of the iterative-compression algorithm for OCT due to Reed, Smith, and Vetta~\cite{DBLP:journals/orl/ReedSV04}.

\begin{subclaim}\label{claim:oct:separator}
For each partitioned OCT~$W = (W_L, W_I)$ of~$G$, for each partition of~$W_L = W_{L,1} \cup W_{L,2}$ into two independent sets, for each proper 2-coloring~$c$ of~$G - W$, we have the following equivalence for each~$X \subseteq V(G) \setminus W$: the graph~$(G - W_I) - X$ has a proper 2-coloring with~$W_{L,1}$ color~$1$ and~$W_{L,2}$ color~$2$ \emph{if and only if} the set~$X$ separates~$A$ from~$R$ in the graph~$G - W$, with:
\begin{align*}
A &= (N_{G-W_I}(W_{L,1}) \cap c^{-1}(1)) \cup (N_{G-W_I}(W_{L,2}) \cap c^{-1}(2)) \\
R &= (N_{G-W_I}(W_{L,1}) \cap c^{-1}(2)) \cup (N_{G-W_I}(W_{L,2}) \cap c^{-1}(1)).
\end{align*}
Observe that~$c^{-1}(i) \subseteq V(G - W)$ for each~$i \in [2]$, so that~$A \cup R \subseteq V(G - W)$, and that the separator~$X$ is allowed to intersect~$A \cup R$.
\end{subclaim}
\begin{claimproof}
($\Rightarrow$) Suppose that~$(G - W_I) - X$ has a proper 2-coloring with~$W_{L,1}$ color~$1$ and~$W_{L,2}$ color~$2$. Suppose for a contradiction that~$X$ is not an~$(A,R)$-separator in~$G-W$, that is, in~$(G-W)-X$ there is a connected component~$H$ simultaneously containing a vertex~$a \in A$ and a vertex~$r \in R$. Note that~$H$ is also a connected subgraph of~$G-W$ and therefore bipartite, which means that if~$|V(H)| \geq 2$ there is a unique partition of~$H$ into two independent sets, so that~$H$ has exactly two proper $2$-colorings depending on which independent set is called color~$1$ and which is called color~$2$. Note that if~$|V(H)| = 1$, the fact that~$H$ has exactly two proper $2$-colorings is trivial. It follows that any proper $2$-coloring of~$H$ either coincides with the 2-coloring~$c$ of~$G-W$, or is such that every vertex gets the opposite of its current color under~$c$.

The fact that~$a \in A$ means by definition that either we have~$c(a) = 1$ and~$a$ is adjacent to a vertex of~$W_{L,1}$, or~$c(a) = 2$ and~$a$ is adjacent to a vertex of~$W_2$. In either case, it shows that in any proper~$2$-coloring of~$(G - W_I) - X$ in which~$W_{L,1}$ gets color~$1$ and~$W_{L,2}$ gets color~$2$, the color of~$a$ must be different from its color under~$c$. By an analogous argument, the fact that~$r \in R$ means that in any proper $2$-coloring of~$(G - W_I) - X$ in which~$W_{L,1}$ gets color~$1$ and~$W_{L,2}$ gets color~$2$, the color of~$r$ must be identical to its color under~$c$.

Since~$a$ and~$r$ belong to the same connected subgraph~$H$ of~$(G - W_I) - X$, in any proper $2$-coloring they either both change their color compared to~$c$, or both keep their color compared to~$c$. This is a contradiction to the fact that~$a$ changed color and~$r$ remained of the same color.

($\Leftarrow$) For the converse, consider a set~$X$ that separates~$A$ from~$R$ in~$G - W$. We construct a proper $2$-coloring~$c'$ of~$(G - W_I) - X$ in which~$W_{L,1}$ gets color~$1$ and~$W_{L,2}$ gets color~$2$, as follows. Let~$c'(v \in W_{L,1}) = 1$ and~$c'(v \in W_{L,2}) = 2$. For each connected component of~$(G - W) - X$ that contains a vertex from~$R$, let its coloring under~$c'$ be identical to its coloring under~$c$. For each connected component of~$(G - W) - X$ that contains no vertex from~$R$, let its coloring under~$c'$ be the opposite of its coloring under~$c$. Since~$c$ was a proper coloring, there are no color conflicts among vertices of~$(G - W) - X$. Since both~$W_{L,1}$ and~$W_{L,2}$ are independent sets, there are no color conflicts among~$W_{L,1}$ or among~$W_{L,2}$. It remains to verify that each edge connecting~$W_L$ to a vertex of~$(G - W) - X$ is properly colored. But this follows from our construction: all neighbors of~$W_{L,1}$ with color~$1$ under~$c$ belong to~$A$ and therefore have their coloring swapped to~$2$ in~$c'$; similarly all neighbors of~$W_{L,2}$ with color~$2$ under~$c$ belong to~$A$ and have their coloring swapped to~$1$ in~$c'$. Finally, neighbors of~$W_{L,1}$ with color~$2$ in~$c$ belong to~$R$ and therefore have the same color~$2$ in~$c'$, and neighbors of~$W_{L,2}$ with color~$1$ in~$c$ belong to~$R$ and have the same color~$1$ in~$c'$, ensuring these edges are properly colored as well.
\end{claimproof}

The next two claims show that certain types of OCTs can be computed efficiently in the~$(s(k),2k)$-unbreakable input graph~$G$.

\begin{subclaim}\label{claim:weaklyconsistentocts}
There is an FPT algorithm that outputs a list of partitioned OCTs in~$G$ with the guarantee that for each witness~$X$, there is a partitioned OCT on the list that is weakly consistent with~$X$.
\end{subclaim}
\begin{claimproof}
The algorithm proceeds as follows.

\begin{enumerate}
    \item Initialize an empty list $\mathcal{W}$. Compute a minimum cardinality odd cycle transversal $W \subseteq V(G)$ of size at most $s(k)+k-1$. If no such OCT exists, return the empty list.
    \item For each $y \in V(G)$, let $\mathcal{C}_y$ be the set of connected vertex sets $S$ with $y \in S$, $|S| \leq s(k)$ and $|N(S)| \leq k$. For each partition $P = (W_L,W_I,W_R)$ of $W$, a choice tuple $t_P$ contains an entry for each $y \in W_R$, where entry $t_P[y]$ is some set $C_y \in \mathcal{C}_y$.
    \item For each partition $P = (W_L,W_I,W_R)$ of $W$ and each choice tuple $t_P$, if $(W \setminus W_R) \cup \bigcup_{y \in W_R}N(t_P[y])$ is an OCT, then add $(W_L,W_I \cup \bigcup_{y \in W_R}N(t_P[y]))$ to $\mathcal{W}$.
    \item Return the list $\mathcal{W}$.
\end{enumerate}
We argue the running time of the steps described above. The first step can be done in time $\Oh^*(3^{s(k)+k})$~\cite{DBLP:journals/orl/ReedSV04,DBLP:books/sp/CyganFKLMPPS15}. 
For each $y \in V(G)$, computing $\mathcal{C}_y$ is in FPT by~\ref{lem:enumerate_small_connected_graphs}. Since there are $3^{s(k)+k-1}$ possible partitions $P$ and FPT many choice tuples $t_P$, the running time follows. To see the correctness of the algorithm, first note that each partition in the output is an OCT by construction. All that is left to show is the output guarantee. Consider some witness $X$ and $C$ be the unique large component of $G-X$. Let $P = (W_L,W_I,W_R)$ be the partition such that $W \cap C = W_L$, $W \cap X = W_I$, and $W_R \subseteq V(G) \setminus (X \cup C)$. To see that $|W_L| \leq k$, observe that if this was not the case we would obtain a smaller OCT by taking $(W \setminus W_L) \cup N(C)$, contradicting $W$ has minimum cardinality. For each $y \in W_R$, let $C_y$ be the connected component of $G-X$ containing $y$. Note that $|C_y| \leq s(k)$ by Proposition~\ref{prop:nice} and $|N(C_y)| \leq k$ by Note~\ref{note:smallneighborhood}. Therefore for some choice tuple $t_P$ we have $t_P[y] = C_y$ for each $y \in W_R$. It follows that $(W_L,W_I \cup \bigcup_{y \in W_R}N(t_P[y]))$ that is contained in the list satisfies the output requirement for witness $X$.
\end{claimproof}

\begin{subclaim}\label{claim:stronglyconsistent}
There is an FPT algorithm that, given a partitioned OCT that is weakly consistent with some (unknown) witness~$X$ in~$G$, outputs a list of OCTs in~$G$ such that at least one is strongly consistent with~$X$.
\end{subclaim}
\begin{claimproof}
Let $(W_L,W_I)$ be the given partitioned OCT, where $W_L \cup W_I = W$. If~$|W| > s(k) + k - 1$, then no witness is strongly consistent with~$W$ by Proposition~\ref{prop:nice}, hence we may assume $|W| \leq s(k)+k-1$.
\begin{enumerate}
    \item Initialize an empty list $\mathcal{W}$. For each $y \in V(G)$, let $\mathcal{C}_y$ be the set of connected vertex sets $S$ with $y \in S$, $|S| \leq s(k)$ and $|N(S)| \leq 2k$.
    \item Let $c^*$ be an arbitrary proper 2-coloring of $G-W$ and let $B_i^* = (c^*)^{-1}(i)$ for each $i \in [2]$. 
    \item For each partition $(W_1,W_2)$ of $W_L$, let $B_1 = N(W_2) \setminus W$ and $B_2 = N(W_1) \setminus W$. Let $A = (B_1 \cap B_2^*) \cup (B_2 \cap B_1^*)$ and $R = (B_1 \cap B_1^*) \cup (B_2 \cap B_2^*)$. 
    \item For each choice $Q \in \{A,R\}$ with $|Q| \leq s(k)+k$, for each $D \subseteq Q$ with $|D| \leq k$, choice tuple $t_{Q,D}$ has an entry for each $y \in Q \setminus D$, where entry $t_{Q,D}[y]$ is some vertex set $C_y \in \mathcal{C}_y$. 
    \item For each choice $Q \in \{A,R\}$ with $|Q| \leq s(k)+k$, for each $D \subseteq Q$ with $|D| \leq k$, and for each choice tuple $t_{Q,D}$, add $(W \cup D \cup \bigcup_{y \in Q \setminus D}N(t_{Q,D}[y])) \setminus W_L$ to $\mathcal{W}$ in case it is an OCT.
    \item Return the list $\mathcal{W}$.
\end{enumerate}
The running time follows from Lemma~\ref{lem:enumerate_small_connected_graphs} and the fact that there are FPT many choices for $(W_1,W_2)$, $D$, and tuple $t_{Q,D}$. We argue the correctness of the algorithm.
Note that each set in the output list is an OCT by construction. Consider some witness $X$ with $(W_L,W_I)$ weakly consistent with $X$ and let $C$ be the unique large component of $G-X$, which is bipartite by definition of witness. Let $Y \subseteq X$ be an OCT of $G$ with $W_I \subseteq Y$ and $Y \subseteq W_I \cup N(C)$. Note that such an OCT $Y$ exists as $W' = (W \setminus W_L) \cup N(C)$ is such an OCT. Let $c \colon V(G) \setminus Y \to [2]$ be a proper 2-coloring of $G-Y$. For some partition $(W_1,W_2)$ of $W_L$ we have $W_i \subseteq c^{-1}(i)$ for each $i \in [2]$. Note that since $W \setminus W_L = W_I$, we have that $|W' \setminus W_I| \leq k$.      

%\bmpr{Change to Claim~\ref{claim:oct:separator}, and also apply in converse direction later in the proof}
By Claim~\ref{claim:oct:separator}, it follows that $Y \setminus W_I \subseteq N(C)$ separates $A$ and $R$ in $G-W$. Note that $B_i \subseteq N[C]$ for each $i \in [2]$ since $W_L \subseteq C$, therefore $A \subseteq N[C]$ and $R \subseteq N[C]$. Observe that $W_L \cup N(C)$ is an $(A,R)$-separator of size at most $2k$ in $G$. Therefore, since $G$ is $(s(k),2k)$-unbreakable, it follows that at least one of the two sides has size at most $s(k)$ after deleting $W_L \cup N(C)$. Let $Q \in \{A,R\}$ be the small side, the algorithm tries this choice as $|Q| \leq s(k)+k$ is satisfied. Let $D = N(C) \cap Q$. 
For each $y \in Q \setminus D$, let $C_y$ be the connected component of $G-(N(C) \cup W_L)$ containing $y$. Note that $|C_y| \leq s(k)$ and $|N(C_y)| \leq 2k$. Let the choice tuple $t_{Q,D}$ be such that $t_{Q,D}[y] = C_y$ for each $y \in Q \setminus D$. Observe that $(D \cup \bigcup_{y \in Q \setminus D}N(t_{Q,D}[y])) \setminus W_L \subseteq N(C)$ is an $(A,R)$-separator in $G-W$. Therefore $(W_I \cup D \cup \bigcup_{y \in Q \setminus D}N(t_{Q,D}[y])) \setminus W_L$ is an OCT by Claim~\ref{claim:oct:separator} contained in $X$, concluding the proof.
\end{claimproof}

With the two claims above, we can solve the problem as follows. Compute a list of partitions $\mathcal{W}$ using Claim~\ref{claim:weaklyconsistentocts} and use each $W \in \mathcal{W}$ as input to Claim~\ref{claim:stronglyconsistent}. Using the output $\mathcal{U}$ of Claim~\ref{claim:stronglyconsistent}, call Lemma~\ref{lem:enumeration_find_witness} for each $U \in \mathcal{U}$. By the output guarantee of the claims, for each witness $X$ we call the lemma with $U \subseteq X$ at some point, thus solving the problem.
\end{proof}

\section{Conclusion}
We have shown that $\hh$-elimination distance and $\hh$-treewidth are non-uniformly fixed-parameter tractable for $\hh$ being the class of bipartite graphs, and whenever~$\hh$ is defined by a finite set of forbidden induced subgraphs. % The main requirement is that enumerating all minimal $\hh$-deletion sets of size at most $k$ can be done in FPT time. 
An obvious direction for further research is extending this to other graph classes. While the algorithms presented here solve the decision variant of the problem, by self-reduction they can be used to identify a witness if one exists. The main observation driving such a self-reduction is the following: if~$\hhtw(G) \leq k$, then for an arbitrary~$v \in V(G)$ there exists a $\hhtw(G)$-witness that contains~$v$ if and only the graph~$G'$ obtained from~$G$ by inserting a minimal forbidden induced subgraph into~$\hh$ and identifying one of its vertices with~$v$, still satisfies~$\hhtw(G') \leq k$. Hence an iterative process can identify all vertices of a witness in this way.

While we have focused on the established notions of~$\hhtw$ and~$\hhdepth$, the ideas presented here can be generalized using minor-closed graph parameters~$f$ other than treewidth and treedepth. As long as~$f$ can attain arbitrarily large values, implying its value on a clique grows with the size of the clique, and~$\hh$ is characterized by a finite set of forbidden induced subgraphs, we believe our approach can be generalized to answer questions of the form: does~$G$ have an $\hh$-deletion set~$X$ for which~$f(\torso_G(X)) \leq k$? % Lemma~\ref{lem:cmsoformulas} gives a CMSO-characterization of the yes-instances, while the approach of Lemma~\ref{lem:enumeration_find_witness} only relied on the fact that the neighborhood size of the large component (which becomes a clique in the torso) has size bounded in~$k$.

%We note that for the case of chordal graphs we would need a different approach: consider the graph obtained by taking the join of $C_n$ and $K_n$. For large enough $n$, this graph is $(s(k),c(k))$-unbreakable for any $s, c \colon \mathbb{N} \to \mathbb{N}$. Each vertex in $V(C_n)$ is in itself a minimal chordal deletion set and hence the number of minimal deletion sets is not bounded in terms of a function of $k$. \jjh{Actually this example works for any graph class whose forbidden subgraph characterization is not finite and is closed under adding universal vertices (interval, perfect, AT-free ...). Is this actually an issue, the nr of minimal solutions does not have to be a function of $k$ per se, as long as they can be computed in FPT time. Can be transformed to show that no FPT bound, take $p C_{n/p}$ joined with $K_{c(k)+1}$ for some $p > k$}

%Furthermore for any minor-closed graph class, such as planar graphs, the set of graphs with $\hhtw(G) \leq k$ or $\hhdepth(G) \leq k$ is minor-closed too, and hence non-uniform FPT algorithms easily follows.

\bibliographystyle{plainurl}
\bibliography{refs}

\clearpage

\appendix

%As an example, the following formula verifies that a subset of vertices $X$ induces a connected graph in $G$.

%\begin{align*}
%    \textsc{conn}(X) = & \forall Y \subseteq V(G) (( \exists u \in X : u \in Y \wedge \exists v\in X : v \notin Y)  \Rightarrow\\
%    & (\exists e \in E \exists u \in X \exists v \in X : \inc(u,e) \wedge \inc(v,e) \wedge u \in Y \wedge v \notin Y))
%\end{align*}

\section{Proof of Theorem \ref{thm:blackbox}}
Since we slightly changed the statement of Theorem~\ref{thm:blackbox} compared to its original form, we state its proof as given in the full version by Lokshtanov et al.\cite{DBLP:journals/corr/abs-1802-01453} for completeness. We require the following theorem from their paper.
\begin{theorem}\cite[Theorem 22]{DBLP:conf/icalp/LokshtanovRSZ18}\label{thm:reduce_to_unbreakable}
Let $\varphi$ be a CMSO formula. For all $c \in \mathbb{N}$, there exists $s \in \mathbb{N}$ such that if there exists an algorithm that solves $\textsc{CMSO}[\varphi]$ on $(s,c)$-unbreakable structures in time $\Oh(n^d)$ for some $d > 4$, then there exists an algorithm that solves $\textsc{CMSO}[\varphi]$ on general structures in time $\Oh(n^d)$.
\end{theorem}

\begin{proof}[Proof of Theorem~\ref{thm:blackbox}]
Let $\hat{c} \colon \mathbb{N}_0 \to \mathbb{N}_0$ and define $\hat{s} \colon \mathbb{N}_0 \to \mathbb{N}_0$ as follows. For all $k \in \mathbb{N}_0$, let $\hat{s}(k)$ be the constant $s$ in Theorem~\ref{thm:reduce_to_unbreakable} and $c = \hat{c}(k)$. Suppose that $\textsc{CMSO}[\hat{\varphi}]$ is FPT on  $(\hat{s}(k),\hat{c}(k))$-unbreakable structures. Then for every fixed $k$ we can solve it in $\Oh(n^d)$ time for some fixed $d > 4$. By Theorem~\ref{thm:reduce_to_unbreakable} it follows that we can solve $\textsc{CMSO}[\hat{\varphi}]$ in $\Oh(n^d)$ time for every fixed $k$ on general structures. Therefore we can solve it in FPT time on general structures.
\end{proof}

\end{document}